\documentclass[submission,copyright,creativecommons]{eptcs}
\providecommand{\event}{GANDALF 2018} % Name of the event you are submitting to
\usepackage{breakurl}             % Not needed if you use pdflatex only.
\usepackage{underscore}           % Only needed if you use pdflatex.
\usepackage[utf8x]{inputenc}
\usepackage{verbatim}
\usepackage{amssymb}

\usepackage{amsthm}
\usepackage{thmtools,thm-restate}%teste
\usepackage{enumitem}
\usepackage{graphicx}
\usepackage{algorithm}
\usepackage[noend]{algpseudocode}

\newtheorem{defi}{Definition}
\newtheorem{exam}{Example}
\newtheorem{lem}{Lemma}
\newtheorem{teo}{Theorem}

\title{On Finding a First-Order Sentence Consistent with a Sample of Strings}
\author{Thiago Alves Rocha
\institute{Department of Computing\\
Federal Institute of Ceará, Brazil}
\email{thiago.alves@ifce.edu.br}
\and
Ana Teresa Martins\thanks{This author was partially supported by the Brazilian National Council for Scientific and Technological Development (CNPq) under the grant number 424188/2016-3.}
\institute{Department of Computing\\
Federal University of Ceará, Brazil}
\email{ana@dc.ufc.br}
\and
Francicleber Martins Ferreira
\institute{Department of Computing\\
Federal University of Ceará, Brazil}
\email{fran@lia.ufc.br}
}
\def\titlerunning{On Finding a First-Order Sentence Consistent with a Sample of Strings}
\def\authorrunning{T.A. Rocha, A.T. Martins \& F.M. Ferreira}
\begin{document}
\maketitle

\begin{abstract}
We investigate the following problem: given a sample of classified strings, find a first-order sentence of minimal quantifier rank that is consistent with the sample. We represent strings as successor string structures, that is, finite structures with unary predicates to denote symbols in an alphabet, and a successor relation. We use results of the Ehrenfeucht–Fraïssé game over successor string structures in order to design an algorithm to find such sentence. We use conditions characterizing the winning strategies for the Spoiler on successor strings structures in order to define formulas which distinguish two strings. Our algorithm returns a boolean combination of such formulas.
\end{abstract}

\section{Introduction}

In this paper, we explore the problem of finding a first-order formula that describes a given sample of classified strings. This problem is meaningful because strings may be used to model sequences of symbolic data such as biological sequences. For instance, in Table~\ref{tableExample1}, we present a sample of classified strings.

\begin{table}[h!]
	\caption{}
	\begin{center}
		\begin{tabular}{ c c c c }
			\hline
			String & Class\\
			\hline 
			$stviil$ & positive \\  
			$ktvive$ & negative  \\
			$stviie$ & positive  \\
			$stpiie$ & negative \\ 
			\hline
		\end{tabular}
	\end{center}
	\label{tableExample1}
\end{table}

The sample in Table~\ref{tableExample1} represents biological sequences which have been associated with a group of diseases called amyloidosis \cite{wieczorek2014induction}. The first-order sentence below represents that $stv$ occurs in a string, and it describes the sample. Variables range over positions in strings, $P_a(i)$ is true if the symbol $a$ occurs in position $i$, and $S$ represents the successor relation over positions.
$$\begin{array}{ll}
\exists x_1 \exists x_2 \exists x_3 (P_s(x_1) \wedge S(x_1, x_2) \wedge P_t(x_2) \wedge S(x_2, x_3) \wedge P_v(x_3)).
\end{array}$$

An algorithm to deal with the problem of finding a formula of minimal quantifier rank consistent with a given sample of structures over an arbitrary vocabulary is introduced in \cite{kaiser12lbg}. As this algorithm works for arbitrary finite relational structures, it runs in exponential time. This algorithm is applied in a general system for learning formulas defining board game rules. These results are also used in finding reductions and polynomial-time programs \cite{jordan2013experiments,jordan2013learning}. The work in \cite{distingqrank2018} investigates a variation of the problem introduced in \cite{kaiser12lbg} when the class of structures is fixed. This study in \cite{distingqrank2018} considers monadic structures, equivalence structures, and disjoint unions of linear orders.

In this paper, we study the problem introduced in \cite{kaiser12lbg} when the sample consists of strings represented by finite structures with a successor relation and a finite number of pairwise disjoint unary predicates. We call such a structure a successor string structure. A sample $S = (P, N)$ consists of two disjoint, finite sets $P, N$ of successor string structures. Given a sample $S$, the problem is to find a first-order sentence $\varphi_S$ of minimal quantifier rank that is consistent with $S$, i.e., it holds in all structures in $P$ and does not hold in any structure in $N$. The size of the sample is the sum of the lengths of all strings in the sample. We intend to solve this problem in polynomial time in the size of $S$.

Ehrenfeucht–Fraïssé games (EF games) \cite{ehrenfeucht1961} is a fundamental technique in finite model theory \cite{Ebbing95FMT,Gradel05FMT} in proving the inexpressibility of certain properties in first-order logic. For instance, first-order logic cannot express that a finite structure has even cardinality. The Ehrenfeucht–Fraïssé game is played on two structures by two players, the Spoiler and the Duplicator. If the Spoiler has a winning strategy for $k$ rounds of such a game, it means that the structures can be distinguished by a first-order sentence $\varphi$ whose quantifier rank is at most $k$, i.e., $\varphi$ holds in exactly one of these structures.

Besides providing a tool to measure the expressive power of a logic, Ehrenfeucht–Fraïssé games allow one to investigate the similarity between structures. In \cite{montanari05}, explicit conditions are provided for the characterization of winning strategies for the Duplicator on successor string structures. Using these conditions, the minimum number of rounds such that the Spoiler has a winning strategy in a game between two such structures can be computed in polynomial time in the size of the structures. This allows one to define a notion of similarity between successor string structures using Ehrenfeucht–Fraïssé games.

An essential part of the algorithm in \cite{kaiser12lbg} is the computation of $r$-Hintikka formulas from structures. An $r$-Hintikka formula is a formula obtained from a structure $\mathcal{A}$ and a positive integer $r$ that describes the properties of $\mathcal{A}$ on the Ehrenfeucht–Fraïssé game with $r$ rounds \cite{Ebbing95FMT}. An $r$-Hintikka formula $\varphi^r_{\mathcal{A}}$ has size exponential in the size of $\mathcal{A}$ and holds exactly on all structures $\mathcal{B}$ such that the Duplicator has a winning strategy for the Ehrenfeucht–Fraïssé game with $r$ rounds on $\mathcal{A}$ and $\mathcal{B}$. Besides, Hintikka formulas are representative because any first-order formula is equivalent to a disjunction of Hintikka formulas.

We use results of the Ehrenfeucht–Fraïssé game over successor string structures \cite{montanari05} in order to design an algorithm to find a sentence which is consistent with the sample in polynomial time. Also, as the size of a Hintikka formula is exponential in the size of a given structure, our algorithm does not use Hintikka formulas. In our case, we define what we call distinguishability formulas. They are defined for two successor string structures $u$, $v$ and a natural number $r$ based on conditions characterizing the winning strategies for the Spoiler on successor strings structures \cite{montanari05}. In this way, we show that distinguishability formulas hold on $u$, do not hold on $v$, and they have quantifier rank at most $r$. This result is also crucial for the definition of our algorithm and to guarantee its correctness. Our algorithm returns a disjunction of conjunctions of distinguishability formulas. We also show that any first-order formula over successor string structures is equivalent to a boolean combination of distinguishability formulas. This result suggests that our approach has the potential to find any first-order sentence.

Our framework is close to grammatical inference. Research in this area investigates the problem of finding a language model of minimal size consistent with a given sample of strings \cite{de2010grammatical}. A language model can be a deterministic finite automaton (DFA) or a context-free grammar, for instance. Grammatical inference has applications in many areas because strings may be used to model text data, traces of program executions, biological sequences, and sequences of symbolic data in general.

A recent model-theoretic approach to grammatical inference is introduced in \cite{strother2017using}. In this approach, it is also used successor string structures to represent strings and first-order sentences as a representation of formal languages. Then, our approach may also be seen as a model-theoretic framework to grammatical inference. The first main difference is that we work with full first-order logic, while the approach in \cite{strother2017using} uses a fragment called CNPL. Formulas of CNPL have the form $\bigwedge_{w_1 \in W_1} \phi_{w_1} \wedge \bigwedge_{w_2 \in W_2} \phi_{w_2}$ such that each $\phi_u$ is a first-order sentence which defines exactly all strings $w$ such that $u$ is a substring of $w$. Also, CNPL is less expressive than first-order logic. Second, given $k$, the goal of the framework in \cite{strother2017using} is to find a CNPL formula such that $max \{ |w| \mid w \in W_1 \cup W_2 \} \leq k$ and $|w|$ is the length of $w$. Our goal is to find a first-order sentence of minimal quantifier rank.

It is well known that a language is definable in first-order logic over successor string structures if and only if it is a locally threshold testable (LTT) language \cite{thomas1982classifying}. A language is LTT if membership of a string can be tested by inspecting its prefixes, suffixes, and infixes up to some length, and counting infixes up to some threshold. The class of LTT languages is a subregular class, i.e., a subclass of the regular languages \cite{rogers2013cognitive}. A grammatical inference algorithm that returns a DFA may return an automaton which recognizes a language not in LTT. Therefore, our results can be useful when one desires to find a model of an LTT language from a sample of strings. We believe that this is the first work on finding a language model of LTT languages from positive and negative strings.

A recent logical framework to find a formula given a sample, also with a model-theoretic approach, can be found in \cite{mGroheRitzert2017FO,groheRitzert2017MSO}. In this framework, a sample consists of classified elements from only one structure. The problem is to find a hypothesis consistent with the classified elements where this hypothesis is a formula from some logic. Recall that, in our framework, samples consist of many classified structures. Another logical framework for a similar problem is Inductive Logic Programming (ILP) \cite{muggleton1991inductive,de2008logical}. ILP uses logic programming as a uniform representation for the sample and hypotheses. As far as we know, our work has no direct relationship with these frameworks.

This paper is organized as follows. In Section~\ref{sec2}, we give the necessary definitions of formal language theory and finite model theory used in this paper. Also in Section~\ref{sec2}, we have an EF game characterization on strings, and, in Section~\ref{sec3}, we translate it into first-order sentences. In Section~\ref{sec3}, we also introduce the concept of distinguishability formulas providing some useful properties. In Section~\ref{sec4}, we introduce our algorithm, give an example of how the algorithm works, and show its correctness. Furthermore, in this section, we briefly discuss how to find a formula with the minimum number of conjunctions. We conclude in Section~\ref{sec5}.

\section{Formal Languages and EF Games on Strings}\label{sec2}

We consider strings over an alphabet $\Sigma$. The set of all such finite strings is denoted by $\Sigma^*$, and the empty string by $\epsilon$. If $w$ is a string, then $|w|$ is the length of $w$. Let $uv$ denote the concatenation of strings $u$ and $v$. For all $u$, $v$, $w$, $x \in \Sigma^*$, if $w = uxv$, then $x$ is a substring of $w$. Moreover, if $u = \epsilon$ (resp. $v = \epsilon$) we say that $x$ is a prefix (resp. suffix) of $w$. We denote the prefix (resp. suffix) of length $k$ of $w$ by $pref_k(w)$ (resp. $suff_k(w)$). Let $i$ and $j$ be positions in a string. The distance between $i$ and $j$, denoted by $d(i, j)$, is $|i - j|$. A formal language is a subset of $\Sigma^*$. A language is locally threshold testable (LTT) if it is a boolean combination of languages of the form $\{ w \mid u$ is prefix of $w \}$, for some $u \in \Sigma^*$, $\{ w \mid u$ is suffix of $w \}$, for some $u \in \Sigma^*$, and $\{ w \mid w$ has $u$ as infix at least $d$ times $\}$, for some $u \in \Sigma^*$ and $d \in \mathbb{N}$ \cite{zeitoun2014separation}. Therefore, membership of a string can be tested by inspecting its prefixes, suffixes and infixes up to some length, and counting infixes up to some threshold. We assume some familiarity with formal languages. See \cite{hopcroft2006} for details.

We view a string $w = a_1...a_n$ as a logical structure $\mathcal{A}_w$ over the vocabulary $\tau = \{ S, (P_a)_{a \in \Sigma}, min, max \}$ with domain $A = \{1, ..., n\}$, that is, the elements of $A$ are positions of $w$. The predicate $S$ is the successor relation and each $P_a$ is a unary predicate for positions labeled with $a$. The constants $min$ and $max$ are interpreted as the positions $1$ and $n$, respectively. We call these structures successor string structures. We assume some familiarity with first-order logic ($FO$), and we use this logic over successor string structures. For details on first-order logic see \cite{Ebbing95FMT,ebbinghaus2013mathematical}. The size of a first-order formula $\varphi$ is the number of symbols occurring in $\varphi$. By the quantifier rank of a formula, we mean the depth of nesting of its quantifiers as in the following.

\begin{defi}[Quantifier Rank]
	Let $\varphi$ be a first-order formula. The quantifier rank of $\varphi$, written $qr(\varphi)$, is defined as
	$$qr(\varphi):=\left\{\begin{array}{llll}
	0,&\textrm{ if }\varphi\textrm{ is atomic }\\
	max(qr(\varphi_1), qr(\varphi_2)),&\textrm{ if }\varphi = \varphi_1 \square \varphi_2\textrm{ such that }\square \in \{\wedge , \vee, \leftarrow\}\\%n > 1
	qr(\psi),&\textrm{ if }\varphi = \neg \psi\\%n > 1
	qr(\psi) + 1,&\textrm{ if }\varphi = Qx \psi\textrm{ such that }Q \in \{\exists, \forall\}
	\end{array}\right.$$
\end{defi}

\noindent Given a first-order sentence $\varphi$ over successor string structures, the formal language defined by $\varphi$ is simply $L(\varphi) := \{w \in \Sigma^* \mid \mathcal{A}_w \models \varphi \}$. In general, we do not distinguish between successor string structures and strings. As an example, if $\varphi = \exists x P_a(x)$, then $L(\varphi) = \Sigma^*a\Sigma^*$. LTT languages can be defined in terms of first-order logic. A language is definable by a sentence of $FO$ over successor string structures if and only if it is LTT \cite{thomas1982classifying}.

Now, we can formally define the problem we are interested in. A sample $S = (P, N)$ is a finite number of classified strings consisting of two disjoint, finite sets $P, N \subseteq \Sigma^*$ of strings over an alphabet $\Sigma$. Intuitively, $P$ contains positively classified strings, and $N$ contains negatively classified strings. The size of a sample $S$ is the sum of the lengths of all strings it includes. We use $|S|$ to denote the size of the sample $S$. A sentence $\varphi$ is consistent with a sample $S$ if $P \subseteq L(\varphi)$ and $N \cap L(\varphi) = \emptyset$. Therefore, a sentence is consistent with a sample if it holds in all strings in $P$ and does not hold in any string in $N$. Given a sample $S$, the problem consists of finding a first-order sentence $\varphi$ of minimum quantifier rank such that $\varphi$ is consistent with $S$.

It is well known that every finite structure can be characterized in first-order logic up to isomorphism, i.e., for every finite structure $\mathcal{A}$, there is a first-order sentence $\varphi_{\mathcal{A}}$ such that for all structures $\mathcal{B}$ we have $\mathcal{B} \models \varphi_{\mathcal{A}}$ iff $\mathcal{A}$ and $\mathcal{B}$ are isomorphic. Since samples are finite sets of finite structures, one can easily build in polynomial-time a first-order sentence consistent with a given sample. For example, let $P = \{ bbabbb, baba \}$ and $N = \{ bbbb \}$. The sentence $\varphi_{bbabbb} \vee \varphi_{baba}$ is consistent with the sample. Unfortunately, the quantifier rank of $\varphi_{\mathcal{A}}$ is the number of elements in the domain of $\mathcal{A}$ plus one. Then, $\varphi = \exists x P_a(x)$ is also consistent with the sample and $qr(\varphi) < qr(\varphi_{bbabbb} \vee \varphi_{baba})$. Therefore, $\varphi_{bbabbb} \vee \varphi_{baba}$ is not a solution to the problem.

Now, we focus on Ehrenfeucht–Fraïssé games and its importance in order to solve the problem we are considering. Let $r$ be an integer such that $r \geq 0$, $u$ and $v$ two successor string structures. The Ehrenfeucht–Fraïssé game $\mathcal{G}_{r}(u,v)$ is played by two players called the Spoiler and the Duplicator. Each play of the game has $r$ rounds and, in each round, the Spoiler plays first and picks an element from the domain $A$ of $u$, or from the domain $B$ of $v$. Then, the Duplicator responds by picking an element from the domain of the other structure. Let $a_i \in A$ and $b_i \in B$ be the two elements picked by the Spoiler and the Duplicator in the $i$th round. The Duplicator wins the play if the mapping $(a_1,b_1), ..., (a_r,b_r)$ is an isomorphism between the substructures induced by $a_1, ..., a_r$ and $b_1,...,b_r$, respectively. Otherwise, Spoiler wins this play. We say that a player has a winning strategy in $\mathcal{G}_{r}(u,v)$ if it is possible for her to win each play whatever choices are made by the opponent. In this work, we always assume that $u$ is different from $v$. Note that if $r \geq |u| + |v|$, then the Spoiler has a winning strategy. Therefore, we assume that $r$ is bounded by $|u| + |v|$. Now, we define formulas describing the properties of a structure in EF games.

\begin{defi}[Hintikka Formulas]
	Let $w$ be a structure, $\overline{a} = a_1...a_s \in A^s$, and $\overline{x} = x_1, ..., x_s$ a tuple of variables, 
	$$\varphi^0_{w,\overline{a}}(\overline{x}) := \bigwedge \{ \varphi(\overline{x}) \mid \varphi \textrm{ is atomic or negated atomic and } w \models \varphi[\overline{a}]\},$$
	$$\textrm{and for } r > 0, \textrm{ } \varphi^r_{w,\overline{a}}(\overline{x}) := \bigwedge_{a \in A} \exists x_{s+1} \varphi^{r-1}_{w,\overline{a}a}(\overline{x},x_{s+1}) \wedge \forall x_{s+1} (\bigvee_{a \in A} \varphi^{r-1}_{w,\overline{a}a}(\overline{x},x_{s+1})).$$
\end{defi}

\noindent A Hintikka formula $\varphi^r_{w,\overline{a}}$ describes the isomorphism type of the substructure generated by $\overline{a}$ in $w$. We write $\varphi^r_{w}$ whenever $s = 0$. Given a string $w$ and a positive integer $r$, the size of the $r$-Hintikka formula $\varphi^r_{w}$ is $O(2^r \times |w|^r)$. Therefore, since $r$ is bounded by $|w|$, the size of $\varphi^r_{w}$ is exponential in the size of $w$. The following theorems are important to prove our main results. They are presented in \cite{Ebbing95FMT} (Theorem~2.2.8 and Theorem~2.2.11).
%\cite[Theorem~3]{referencia}

\renewcommand{\labelitemi}{$\bullet$}

\begin{teo}[Ehrenfeucht's Theorem]\label{teoEFHintikka}%\cite{ehrenfeucht1961application}
	Given $u$ and $v$, and $r \geq 0$, the following are equivalent:
	\begin{itemize}
		\item the Duplicator has a winning strategy in $\mathcal{G}_r(u, v)$.
		\item If $\varphi$ is a sentence of quantifier rank at most $r$, then $u \models \varphi$ iff $v \models \varphi$.
		\item $v \models \varphi^r_{u}$.
	\end{itemize}
\end{teo}

\begin{teo}\label{hintikkaRepresentative}%\cite{ehrenfeucht1961application}
	Let $\varphi$ be a sentence of quantifier rank at most $r$. Then, there exists structures $u_1$, ..., $u_k$ such that
	$$\models \varphi \leftrightarrow (\varphi^r_{u_1} \vee ... \vee \varphi^r_{u_s}).$$
\end{teo}

\noindent We use Theorem~\ref{hintikkaRepresentative} in order to show that any first-order formula over successor string structures is equivalent to a boolean combination of distinguishability formulas. EF games are essential in our framework because if the Spoiler has a winning strategy in a game on strings $u$ and $v$ with $r$ rounds, then there exists a first-order sentence $\varphi$ of quantifier rank at most $r$ that holds in $u$ and does not hold in $v$. Also, in this case, the sentence $\varphi^r_{u}$ is an example of such a sentence. Unfortunately, over arbitrary vocabularies, the problem of determining whether the Spoiler has a winning strategy is $PSPACE$-complete \cite{pezzoli99}.

However, it is possible to do better in the particular case of EF games on successor string structures. For details see \cite{montanari05}. First, we need the following definitions. Let $A \subseteq \mathbb{N}$. A partition of $A$ is a collection of subsets $X$ of $A$ such that each element of $A$ is included in exactly one subset. An $l$-segmentation of $A$ is a partition of $A$ with the minimum number of subsets such that for all $i, j$ in the same subset, $d(i, j) \leq l$ and if $i, j$ are in the same subset $X$ and $i \leq h \leq j$, then $h \in X$. Each subset $X$ in the partition is called a segment.

In the following, we consider substrings $\alpha$ over $\Sigma$ such that $|\alpha| =  2^{q_{\alpha}} - 1$ for some $q_{\alpha} > 0$. Let $w = w_1...w_n$ be a string such that $w_i \in \Sigma$, for $i \in \{1, ..., n \}$. An occurrence of $\alpha$ is centered on a position $i$ in a string $w$ if $w_{i - (2^{q_{\alpha}-1}-1)}...w_i...w_{i + 2^{q_{\alpha}-1}-1} = \alpha$. An occurrence of $\alpha$ centered on a position $i$ in $w$ is free if $|min^{w} - i| > 2^{q_{\alpha} - 1}$ and $|max^{w} - i| > 2^{q_{\alpha} - 1}$. The set of free occurrences of $\alpha$ in $w$ is $\Gamma(\alpha, w) = \{ i \mid |min^{w} - i| > 2^{q_{\alpha} - 1}, |max^{w} - i| > 2^{q_{\alpha} - 1}, w_{i - (2^{q_{\alpha}-1}-1)}...w_i...w_{i + 2^{q_{\alpha}-1}-1} = \alpha \}$. The free multiplicity of $\alpha$ in $w$, denoted by $\gamma(\alpha, w)$, is the number of free occurrences of $\alpha$ in $w$, i.e., $|\Gamma(\alpha, w)|$. The free scattering of $\alpha$ in $w$, denoted by $\sigma(\alpha, w)$, is the number of segments in a $2^{q_{\alpha}}$-segmentation of $\Gamma(\alpha, w)$.

\begin{exam}
	Let $w = ababababbababaaba$ and $\alpha = aba$. Note that $q_{\alpha} = 2$. The occurrence of $\alpha$ centered on position $2$ in $w$ is not free because $|min^w - 2| \leq 2$. However, the occurrence of $\alpha$ centered on position $4$ in $w$ is free. The set of free occurrences of $\alpha$ in $w$ is $\Gamma(aba, w) = \{4, 6, 11, 13\}$. Therefore, $\gamma(\alpha, w) = 4$. A $2^{q_{\alpha}}$-segmentation of $\Gamma(\alpha, w)$ is $\{ \{4, 6\}, \{11, 13 \} \}$. Then, $\sigma(\alpha, w) = 2$.
\end{exam}

\noindent Now, we have a result of EF games on successor string structures.

\begin{teo}\label{EFGamesOnStrings}\cite{montanari05}
	Let $r$ be a natural number, $u$ and $v$ be strings. The Duplicator has a winning strategy in $\mathcal{G}_r(u, v)$ if and only if the following conditions hold:
	\begin{enumerate}%[label=\roman*.]
		\item \label{gameLength} $d(min^u, max^u) = d(min^v, max^v)$ or $(d(min^u, max^u) > 2^r$ and $d(min^v, max^v) > 2^r)$;
		\item \label{gameBorder} $pref_{2^r}(u) = pref_{2^r}(v)$ and $suff_{2^r}(u) = suff_{2^r}(v)$;
		\item \label{gameSub} $\sigma(\alpha, u) + q_{\alpha} > r$ and $\sigma(\alpha, v) + q_{\alpha} > r$ for all $\alpha$ such that $|\alpha| = 2^{q_{\alpha}} - 1$ and $(\sigma(\alpha, u) \neq \sigma(\alpha, v)$ or $\gamma(\alpha, u) \neq \gamma(\alpha, v))$.
	\end{enumerate}
\end{teo}

\noindent Besides the importance of EF games on strings to our framework, we also use the above result to define the distinguishability formulas. These formulas are defined based on the conditions characterizing a winning strategy for the Spoiler on successor string structures. In \cite{montanari05}, this result is also used to define a notion of similarity between successor string structures using Ehrenfeucht–Fraïssé games. The EF-similarity between strings $u$ and $v$, written $EFsim(u, v)$, is the minimum number of rounds $r$ such that the Spoiler has a winning strategy in the game $\mathcal{G}_r(u, v)$. Then, the EF-similarity between two strings can be computed in polynomial time in the size of the strings in the following way.
$$\begin{array}{lllll}
EFsim(u, v) := \begin{array}{ll}
min\{ simLength(u,v), simPref(u, v), simSuff(u, v), simSub(u,v) \},\textrm{ such that}
\end{array}\\
simLength(u, v) = \lceil log_2(min(|u|, |v|) - 1) \rceil,\\
simPref(u, v) = min\{\lceil log_2(k) \rceil \mid pref_k(u) \neq pref_k(v)\},\\
simSuff(u,v) = min\{\lceil log_2(k) \rceil \mid suff_k(u) \neq suff_k(v)\},\\
simSub(u,v) = min\{q_{\alpha} + min(\sigma(\alpha, u), \sigma(\alpha, v)) \mid \gamma(\alpha, u) \neq \gamma(\alpha, v) \textrm{ or } \sigma(\alpha, u) \neq \sigma(\alpha, v) \}.
\end{array}$$

\noindent Given two strings $u$ and $v$, $EFsim(u, v)$ can be computed in $O((|u|+|v|)^2log(|u|+|v|))$, that is, it can be computed in polynomial time \cite{montanari05}. Our algorithm's first step is to compute the sufficient quantifier rank to distinguish between any two strings $u \in P$ and $v \in N$. Then, the fact that $EFsim(u, v)$ can be computed in polynomial time is important to show that our algorithm runs in polynomial time as well.

It is easy to build a first-order sentence consisting of a disjunction of Hintikka formulas of minimal quantifier rank that is consistent with a given sample. For example, let $P = \{u_1\}$, $N = \{v_1, v_2\}$, $r = max\{EFsim(u_1, v_1), EFsim(u_1, v_2)\}$, and $S = (P, N)$. The sentence $\varphi^{r}_{u_1}$ is a first-order sentence of minimal quantifier rank that is consistent with $S$. Unfortunately, the size of $\varphi^{r}_{u_1}$ is exponential in the size of $S$. Therefore, $\varphi^{r}_{u_1}$ can not be built in polynomial time in the size of the sample. This motivates the introduction of distinguishability formulas in Section~\ref{sec3}.

\section{Distinguishability Formulas}\label{sec3}

In this section, we define distinguishability formulas for strings $u$, $v$ and a natural number $r$. Distinguishability formulas are formulas that hold on $u$, do not hold on $v$ and they have quantifier rank at most $r$. The first step is to show that the conditions of Theorem~\ref{EFGamesOnStrings} can be expressed by first-order formulas. These formulas are defined recursively in order to reduce the quantifier rank. The recursive definitions can all be simplified to direct definitions with higher quantifier ranks but, in this case, we can not guarantee that the quantifier rank is adequate. These formulas are also important to help the explanation, and they improve readability of sentences returned by our algorithm.

First, we introduce $\varphi^{d(t_1, t_2)}_{\leq n}$. It describes that the distance between terms $t_1$ and $t_2$ is at most $n$. This can be used to represent condition~\ref{gameLength} of Theorem~\ref{EFGamesOnStrings}.
$$\varphi^{d(t_1, t_2)}_{\leq n}:=\left\{\begin{array}{ll}
t_1 = t_2 \vee S(t_1, t_2),&\textrm{ if }n = 1\\
\exists y (\varphi^{d(t_1, y)}_{\leq \lfloor \frac{n}{2} \rfloor} \wedge \varphi^{d(y, t_2)}_{\leq \lceil \frac{n}{2} \rceil}),&\textrm{ otherwise.}%n > 1
\end{array}\right.$$

\noindent We also set $\varphi^{d(t_1, t_2)}_{> n} := \neg \varphi^{d(t_1, t_2)}_{\leq n}$, $\varphi^{d(t_1, t_2)}_{\geq n} := \varphi^{d(t_1, t_2)}_{> n - 1}$, $\varphi^{d(t_1, t_2)}_{< n} := \neg \varphi^{d(t_1, t_2)}_{\geq n}$, and $\varphi^{d(t_1, t_2)}_{= n} := \varphi^{d(t_1, t_2)}_{\leq n} \wedge \varphi^{d(t_1, t_2)}_{\geq n}$. Clearly, $qr(\varphi^{d(t_1, t_2)}_{\vartriangleright n}) = \lceil log_2(n) \rceil$ and $w \models \varphi^{d(min, max)}_{\vartriangleright n} \textrm{ iff } |w| \vartriangleright n + 1$ for $\vartriangleright$ $\in \{<, >, \leq, \geq, =\}$. Besides, the size of $\varphi^{d(t_1, t_2)}_{\vartriangleright n}$ is $O(n)$. For example, for $r = 3$ and strings $u$ and $v$ such that $|u| = 9$ and $|v| = 12$, we have that $u \models \varphi^{d(min, max)}_{\leq 8}$, $v \not\models \varphi^{d(min, max)}_{\leq 8}$, and $qr(\varphi^{d(min, max)}_{\leq 8}) = 3$. Then, $d(min^u, max^u) \neq d(min^v, max^v)$ and $d(min^u, max^u) \leq 2^r$. Therefore, the Spoiler has a winning strategy for $\mathcal{G}_3(u, v)$.

Now, we turn to the cases in which substrings are important.
These cases are conditions~\ref{gameBorder} and \ref{gameSub} from Theorem~\ref{EFGamesOnStrings}. Formulas $\varphi_{t_1a_1...a_kt_2}$ hold in a string $w$ when the string between $t_1$ and $t_2$ is $a_1...a_k$. Formulas $\varphi_{ta_1...a_k}$ and $\varphi_{a_1...a_kt}$ express that a string $a_1...a_k$ occurs immediately on the right and immediately on the left of a term $t$, respectively.

$$\varphi_{t_1a_1...a_kt_2}:=\left\{\begin{array}{lll}
\exists z (P_{a_1}(z) \wedge S(t_1, z) \wedge S(z, t_2)),&\textrm{ if }k = 1\\
\exists z (P_{a_1}(z) \wedge S(t_1, z) \wedge \varphi_{za_2t_2}),&\textrm{ if }k = 2\\
\exists z (P_{a_{\lceil \frac{k}{2} \rceil}}(z) \wedge \varphi_{t_1a_1...a_{\lceil \frac{k}{2} \rceil - 1}z} \wedge \varphi_{za_{\lceil \frac{k}{2} \rceil + 1}...a_kt_2}),&\textrm{ otherwise.}
\end{array}\right.$$

$$\varphi_{ta_1...a_k}:=\left\{\begin{array}{lll}
\exists y (S(t, y) \wedge P_{a_1}(y)),&\textrm{ if }k = 1\\
\exists y (P_{a_{1}}(y) \wedge S(t, y) \wedge \varphi_{ya_{2}}),&\textrm{ if }k = 2\\
\exists y (P_{a_{\lceil \frac{k}{2} \rceil}}(y) \wedge \varphi_{ta_1...a_{\lceil \frac{k}{2} \rceil - 1}y} \wedge \varphi_{ya_{\lceil \frac{k}{2} \rceil + 1}...a_k}),&\textrm{ otherwise.}
\end{array}\right.$$

$$\varphi_{a_1...a_kt}:=\left\{\begin{array}{lll}
\exists y (S(y, t) \wedge P_{a_1}(y)),&\textrm{ if }k = 1\\
\exists y (P_{a_{1}}(y) \wedge \varphi_{ya_{2}t}),&\textrm{ if }k = 2\\
\exists y (P_{a_{\lceil \frac{k}{2} \rceil}}(y) \wedge \varphi_{a_1...a_{\lceil \frac{k}{2} \rceil - 1}y} \wedge \varphi_{ya_{\lceil \frac{k}{2} \rceil + 1}...a_kt}),&\textrm{ otherwise.}
\end{array}\right.$$

\noindent With respect to the quantifier rank, we have $qr(\varphi_{t_1a_1...a_kt_2}) = qr(\varphi_{ta_1...a_k}) = qr(\varphi_{a_1...a_kt}) = \lceil log_2(k + 1) \rceil$. Furthermore, the size of these formulas is $O(k)$. Now, we define sentences to handle the prefix and suffix of strings. These sentences express that the prefix of length $k$ is $a_1...a_k$ and the suffix of length $k$ is $a_1...a_k$, respectively.

$$\varphi_{pref_k = a_1...a_k}:=\left\{\begin{array}{llll}
P_{a_1}(min),&\textrm{ if }k = 1\\
P_{a_1}(min) \wedge \varphi_{mina_2...a_k},&\textrm{ if }k = 2\textrm{ or }k = 3\\
P_{a_1}(min) \wedge \exists x (P_{a_{\lceil \frac{k+1}{2} \rceil}}(x) \wedge \varphi_{mina_2...a_{\lceil \frac{k+1}{2} \rceil - 1}x} \wedge \varphi_{xa_{\lceil \frac{k+1}{2} \rceil + 1}...a_{k}}),&\textrm{ otherwise.}%k > 3
\end{array}\right.$$

$$\varphi_{suff_k = a_1...a_k}:=\left\{\begin{array}{llll}
P_{a_k}(max),&\textrm{ if }k = 1\\
P_{a_k}(max) \wedge \varphi_{a_1...a_{k-1}max},&\textrm{ if }k = 2\textrm{ or }k = 3\\
P_{a_k}(max) \wedge \exists x (P_{a_{\lfloor \frac{k}{2} \rfloor}}(x) \wedge \varphi_{a_1...a_{\lfloor \frac{k}{2} \rfloor - 1}x} \wedge \varphi_{xa_{\lfloor \frac{k}{2} \rfloor + 1}...a_{k-1}max}),&\textrm{ otherwise.}%k > 3
\end{array}\right.$$

\noindent We also set abbreviations $\varphi_{pref_k \neq a_1...a_k} := \neg \varphi_{pref_k = a_1...a_k}$ and $\varphi_{suff_k \neq a_1...a_k} := \neg \varphi_{suff_k = a_1...a_k}$. Therefore, $qr(\varphi_{pref_k \blacktriangleright a_1...a_k}) = \lceil log_2(k) \rceil$ and $w \models \varphi_{pref_k \blacktriangleright a_1...a_k}$ iff $pref_k(w) \blacktriangleright a_1...a_k$, where $\blacktriangleright$ $\in \{=, \neq\}$. Analogously for $\varphi_{suff_k \blacktriangleright a_1...a_k}$. Also, the size of $\varphi_{pref_k = a_1...a_k}$ and $\varphi_{suff_k = a_1...a_k}$ is $O(k)$. We use these formulas to express condition~\ref{gameBorder} of Theorem~\ref{EFGamesOnStrings}. To see why, Let $r = 2$, $u = bbbaabbbb$ and $v = bbbbabbbb$. Thus, $u \models \varphi_{pref_4 \neq bbbb}$ and $v \not\models \varphi_{pref_4 \neq bbbb}$. Then, $pref_{2^r}(u) \neq pref_{2^r}(v)$ and, from condition~\ref{gameBorder} of Theorem~\ref{EFGamesOnStrings}, it follows that the Spoiler has a winning strategy in $\mathcal{G}_2(u, v)$.

Now, we need sentences regarding free multiplicity and free scattering. Let $\alpha = a_1...a_k$ be a string such that each $a_i \in \Sigma$, and $k = 2^{q_{\alpha}} - 1$ for $q_{\alpha} > 0$ as in condition~\ref{gameSub} from Theorem~\ref{EFGamesOnStrings}. Now, we set the formula $\varphi_{a_1...a_k}(x)$ describing that a string $a_1...a_k$ occurs centered on position $x$. Then, we give an example of a formula $\varphi_{\alpha}(x)$.

$$\varphi_{a_1...a_k}(x):=\left\{\begin{array}{ll}
P_{a_1}(x),&\textrm{ if }k = 1\\
%P_{a_1}(x) \wedge \varphi^{r(x)}_{xa_2},&\textrm{ if }k = 2\\
P_{a_{\lceil \frac{k}{2} \rceil}}(x) \wedge \varphi_{a_1...a_{\lceil \frac{k}{2} \rceil - 1}x} \wedge \varphi_{xa_{\lceil \frac{k}{2} \rceil + 1}...a_k},&\textrm{ if }k \geq 3.
\end{array}\right.$$

\begin{exam}\label{examAlphaCentred}
	Let $\alpha = abc$. Then,%Using the above abbreviation,
	$$\varphi_{\alpha}(x) = \begin{array}{ll}
	P_{b}(x) \wedge \exists y_1 (S(y_1, x) \wedge P_{a}(y_1)) \wedge \exists y_1 (P_{c}(y_1) \wedge S(x, y_1)).
	\end{array}$$
\end{exam}

\noindent Note that $qr(\varphi_{\alpha}(x)) = q_{\alpha} - 1$ and the size of $\varphi_{\alpha}(x)$ is $O(k)$. Now, we can use formulas $\varphi_{\alpha}(x)$ to define $\varphi_{\gamma(\alpha) \geq n}$ expressing that $\alpha$ has at least $n$ free occurrences. Then, we need to use $n$ pairwise different variables and each variable must be in a proper distance from $min$ and $max$.

$$\varphi_{\gamma(\alpha) \geq n} := \begin{array}{ll}
\exists x_1 \exists x_2 ... \exists x_n ( \bigwedge_{1 \leq i < j \leq n} x_i \neq x_j \wedge \bigwedge^{n}_{i = 1} \varphi_{\alpha}(x_i) \wedge \bigwedge^{n}_{i = 1} (\varphi^{d(min, x_i)}_{> 2^{q_{\alpha}-1}} \wedge \varphi^{d(x_i, max)}_{> 2^{q_{\alpha}-1}})).
\end{array}$$

\noindent Now, we need to deal with formulas $\varphi_{\sigma(\alpha) \geq n}$ expressing that the scattering of $\alpha$ is at least $n$. First, in the following, we set an auxiliary formula in order to make the presentation simpler. The formula below indicates that $\alpha$ occurs centered on a position on the left of $x$ and at least $2^{q_{\alpha} - 1}$ distant from $x$. This formula is important in ensuring a proper distance from other occurrences of $\alpha$, that is, greater than $2^{q_{\alpha}}$. Furthermore, the distance between $\alpha$ and $min$ or $max$ must be greater than $2^{q_{\alpha} - 1}$ in order to $\alpha$ occur free.% centred on a free position.

$$\varphi^{d(x) \geq 2^{q_{\alpha}-1}}_{\alpha} := \exists y (\varphi^{d(y, x)}_{\geq 2^{q_{\alpha}-1}} \wedge \varphi_{\alpha}(y) \wedge \varphi^{d(min, y)}_{> 2^{q_{\alpha}-1}} \wedge \varphi^{d(y, max)}_{> 2^{q_{\alpha}-1}}).$$

\noindent With respect to the quantifier rank, we have $qr(\varphi^{d(x) \geq 2^{q_{\alpha}-1}}_{\alpha}) = q_{\alpha}$. Now, we can define the sentence $\varphi_{\sigma(\alpha) \geq n}$. After that, we give an example of $\varphi_{\gamma(\alpha) \geq n}$ and $\varphi_{\sigma(\alpha) \geq n}$.

$$\varphi_{\sigma(\alpha) \geq n} := \begin{array}{ll}
\exists x_1 (\varphi^{d(x_1) \geq 2^{q_{\alpha}-1}}_{\alpha} \wedge \exists x_2 (\varphi^{d(x_1, x_2)}_{> 2^{q_{\alpha}}} \wedge \varphi^{d(x_2) \geq 2^{q_{\alpha}-1}}_{\alpha} \wedge ... \wedge
\exists x_{n-1}\\ (\varphi^{d(x_{n-2}, x_{n-1})}_{> 2^{q_{\alpha}}} \wedge \varphi^{d(x_{n-1}) \geq 2^{q_{\alpha}-1}}_{\alpha} \wedge \exists x_n (\varphi^{d(x_{n-1}, x_n)}_{> 2^{q_{\alpha}-1}} \wedge \varphi_{\alpha}(x_n) ))...)).
\end{array}$$

\begin{exam}
	Let $\alpha = abc$ and $n = 2$. Thus,
	$$\varphi_{\gamma(\alpha) \geq n} = \begin{array}{ll}
	\exists x_1 \exists x_2 (x_1 \neq x_2 \wedge \varphi_{abc}(x_1) \wedge \varphi_{abc}(x_2) \wedge \varphi^{d(min, x_1)}_{> 2^{q_{\alpha}-1}} \wedge \varphi^{d(x_1, max)}_{> 2^{q_{\alpha}-1}} \wedge \varphi^{d(min, x_2)}_{> 2^{q_{\alpha}-1}} \wedge \varphi^{d(x_2, max)}_{> 2^{q_{\alpha}-1}} ).
	\end{array}$$
	
	$$\varphi_{\sigma(\alpha) \geq n} = \begin{array}{ll}
	\exists x_1 (\varphi^{d(x_1) \geq 2^{q_{\alpha}-1}}_{\alpha} \wedge \exists x_2 (\varphi^{d(x_{1}, x_2)}_{> 2^{q_{\alpha}-1}} \wedge \varphi_{\alpha}(x_2) )).
	\end{array}$$
\end{exam}

\noindent We also define the following abbreviations $\varphi_{\gamma(\alpha) < n} := \neg \varphi_{\gamma(\alpha) \geq n}$ and $\varphi_{\gamma(\alpha) = n} := \varphi_{\gamma(\alpha) \geq n} \wedge \varphi_{\gamma(\alpha) < n + 1}$. Then, $qr(\varphi_{\gamma(\alpha) \trianglelefteq n}) = q_{\alpha} + n - 1$ and $w \models \varphi_{\gamma(\alpha) \trianglelefteq n}$ iff $\gamma(\alpha, w) \trianglelefteq n$ for $\trianglelefteq$ $\in \{\geq, <, =\}$. It is analogous to $\varphi_{\sigma(\alpha) \trianglelefteq n}$. Besides, the size of $\varphi_{\gamma(\alpha) \trianglelefteq n}$ and $\varphi_{\sigma(\alpha) \trianglelefteq n}$ is $O((n + |\alpha|)^2)$.

Now, we can define the distinguishability formulas. Distinguishability formulas are defined from a pair of strings $u$, $v$ and a quantifier rank $r$. These formulas have quantifier rank at most $r$, and they hold in $u$ and do not hold in $v$. In what follows, $\alpha$ is a substring of $u$ or $v$.

\begin{defi}[Distinguishability Formulas]
	Let $u$, $v$ be strings over some alphabet $\Sigma$ and $r$ be a natural number. The set of distinguishability formulas from $u$, $v$ and $r$ is
	$$\Phi^r_{u, v}:= \Phi_{u, v}^{r, length} \cup \Phi_{u, v}^{r, pref} \cup \Phi_{u, v}^{r, suff} \cup \Phi_{u, v}^{r, sub}.$$ 
	
	where
	$$\Phi_{u, v}^{r, length}:=\begin{array}{ll}
	\{ \varphi^{d(min, max)}_{\leq n} \mid |u| < |v|, |u| - 1 \leq n \leq min(2^r, |v| - 2) \} \cup\\
	\{ \varphi^{d(min, max)}_{\geq n} \mid |u| > |v|, |v| \leq n \leq min(2^r + 1, |u| - 1) \}
	\end{array}$$
	$$\Phi_{u, v}^{r,pref} := \begin{array}{ll}
	\{ \varphi_{pref_{k} = pref_{k}(u)} \mid pref_{k}(u) \neq pref_{k}(v),  k \leq min(2^r, |u|, |v|)\} \cup\\
	\{ \varphi_{pref_{k} \neq pref_{k}(v)} \mid pref_{k}(u) \neq pref_{k}(v), k \leq min(2^r, |u|, |v|) \}
	\end{array}$$
	$$\Phi_{u, v}^{r,suff} := \begin{array}{ll}
	\{ \varphi_{suff_{k} = suff_{k}(u)} \mid suff_{k}(u) \neq suff_{k}(v),  k \leq min(2^r, |u|, |v|)\} \cup\\
	\{ \varphi_{suff_{k} \neq suff_{k}(v)} \mid suff_{k}(u) \neq suff_{k}(v), k \leq min(2^r, |u|, |v|) \}
	\end{array}$$	
	$$\Phi_{u, v}^{r, sub}:=\begin{array}{llll}
	\{ \varphi_{\sigma(\alpha) \geq n} \mid \sigma(\alpha, u) > \sigma(\alpha, v), \sigma(\alpha, v) < n \leq min(r - q_{\alpha} + 1, \sigma(\alpha, u)) \} \cup\\
	\{ \varphi_{\sigma(\alpha) < n} \mid \sigma(\alpha, u) < \sigma(\alpha, v), \sigma(\alpha, u) < n \leq min(r - q_{\alpha} + 1, \sigma(\alpha, v)) \} \cup\\
	\{ \varphi_{\gamma(\alpha) \geq n} \mid \gamma(\alpha, u) > \gamma(\alpha, v), \gamma(\alpha, v) < n \leq min(r - q_{\alpha} + 1, \gamma(\alpha, u)) \} \cup\\
	\{ \varphi_{\gamma(\alpha) < n } \mid \gamma(\alpha, u) < \gamma(\alpha, v), \gamma(\alpha, u) < n \leq min(r - q_{\alpha} + 1, \gamma(\alpha, v)) \}.
	\end{array}$$
\end{defi}

\noindent Observe that, given $u$, $v$, and $r$, the size of a formula $\varphi \in \Phi^r_{u, v}$ is $O((|u| + |v|)^2)$, and the number of elements $|\Phi^r_{u,v}|$ in $\Phi^r_{u,v}$ is $O((|u| + |v|)^3)$. This is crucial in order to guarantee that our algorithm runs in polynomial time in the size of the sample. Also, by the definition of distinguishability formulas and Theorem~\ref{teoEFHintikka}, it follows that the Spoiler has a winning strategy in $\mathcal{G}_r(u, v)$ if and only if there exists $\varphi \in \Phi^r_{u, v}$. In what follows, we give examples of distinguishability formulas.

\begin{exam}
	Let $u = aaacbbb$, $v = aaabbbbb$, and $r = 2$. Therefore, $\varphi_{pref_4 = aaac}$, $\varphi_{pref_4 \neq aaab} \in \Phi^r_{u, v}$. Furthermore, $\varphi_{pref_3 = aaa}, \varphi_{pref_3 \neq aaa} \not\in \Phi^r_{u,v}$ because $pref_3(u) = pref_3(v)$. Also, $\varphi_{\gamma(c) \geq 1} \in \Phi^r_{u, v}$ because $\gamma(c, u) > \gamma(c, v)$ and $\gamma(c, v) < 1 \leq min(2,1)$. Besides, $\varphi_{\gamma(bbb) < 1} \in \Phi^r_{u, v}$ because $\gamma(bbb, u) < \gamma(bbb, v)$ and $\gamma(bbb, u) < 1 \leq min(1,1)$. With respect to the length, $\Phi^{r, length}_{u,v} = \emptyset$ because $n > min(2^r, 6)$.
\end{exam}

\begin{exam}
	Now, let $u = bbaaaaaaaabb$, $v = bbaaaaaabb$, and $r = 4$. Then, $\varphi_{sigma(aaa) \geq 2} \in \Phi^r_{u, v}$ as $\sigma(aaa,u) = 2$, $\sigma(aaa,v) = 1$, and $\sigma(aaa,v) < n \leq min(3,2)$. Besides, $\varphi^{d(min, max)}_{\geq 10} \in \Phi^r_{u, v}$ because $|u| > |v|$ and $9 \leq 10 \leq min(16,11)$.
\end{exam}

\noindent Now, we show results ensuring adequate properties of distinguishability formulas.

\begin{restatable}{lem}{setConsis}\label{setConsis}
	Let $u, v$ be strings and $r$ be a natural number. Let $\varphi \in \Phi_{u, v}^{r}$. Then, $u \models \varphi$ and $v \not\models \varphi$.
\end{restatable}

\begin{proof}
	First, suppose $\varphi = \varphi^{d(min, max)}_{\leq n}$. Then, $|u| < |v|$ and $|u| - 1 \leq n \leq min(2^r, |v| - 2)$. As $|u| - 1 \leq n$, then $d(min^u, max^u) \leq n$. Therefore, $u \models \varphi$. Clearly, $n \leq |v| - 2$ because $n \leq min(2^r, |v| - 2)$. Then, $n \leq d(min^v, max^v) - 1$. Therefore, $v \not\models \varphi$. The case in which $\varphi = \varphi^{d(min, max)}_{\geq n}$ is similar.
	
	Now, let $\varphi = \varphi_{pref_k = pref_k(u)}$. Then, $pref_k(u) \neq pref_k(v)$. It also holds that $k \leq min(2^r, |u|, |v|)$. As $k \leq |u|$, $k \leq |v|$, and $pref_k(u) \neq pref_k(v)$, then $u \models \varphi$ and $v \not\models \varphi$. The cases in which $\varphi = \varphi_{pref_k \neq pref_k(v)}$, $\varphi = \varphi_{suff_k = suff_k(v)}$, and $\varphi = \varphi_{suff_k \neq suff_k(v)}$ are similar.
	
	Next, let $\varphi = \varphi_{\gamma(\alpha) \geq n}$. Thus, $\gamma(\alpha, v) < \gamma(\alpha, u)$ and $\gamma(\alpha, v) < n \leq min(r - q_{\alpha} + 1, \gamma(\alpha, u))$. Therefore, $\gamma(\alpha, v) < n \leq \gamma(\alpha, u)$. Then, $u \models \varphi$ and $v \not\models \varphi$. The cases in which $\varphi = \varphi_{\gamma(\alpha) < n}$, $\varphi = \varphi_{\sigma(\alpha) \geq n}$, and $\varphi = \varphi_{\sigma(\alpha) < n}$ are analogous.
\end{proof}

\begin{restatable}{lem}{setQR}
	Let $u, v$ be strings and $r$ be a natural number. Let $\varphi \in \Phi_{u, v}^{r}$. Then, $EFsim(u, v) \leq qr(\varphi) \leq r$.
\end{restatable}

\begin{proof}
	From Lemma~\ref{setConsis}, it follows that $EFsim(u, v) \leq qr(\varphi)$. Now, we need to show that $qr(\varphi) \leq r$.
	
	If $\varphi = \varphi^{d(min, max)}_{\trianglelefteq n}$ where $\trianglelefteq$ $\in \{ \leq, \geq \}$, then $n \leq 2^r$. Hence, $qr(\varphi) = \lceil log_2(n) \rceil$. It follows that $qr(\varphi) \leq \lceil log_2(2^r) \rceil = r$.
	
	Let $\varphi \in \{\varphi_{pref_k \blacktriangleright w}, \varphi_{suff_k \blacktriangleright w} \}$ where $\blacktriangleright$ $\in \{=, \neq\}$. Then, $k \leq 2^r$. Therefore, $qr(\varphi) = \lceil log_2(k) \rceil \leq r$.
	
	Finally, if $\varphi \in \{ \varphi_{\gamma(\alpha) \vartriangleright n}, \varphi_{\sigma(\alpha) \vartriangleright n} \}$ where $\vartriangleright$ $\in \{<, \geq\}$, then $n \leq r - q_{\alpha} + 1$. Thus, $qr(\varphi) = q_{\alpha} + n - 1 \leq r$.
\end{proof}

Distinguishability formulas consist of a boolean combination of sentences $\varphi^{d(min, max)}_{\leq n}$, $\varphi_{pref_k = a_1...a_k}$, $\varphi_{suff_k = a_1...a_k}$, $\varphi_{\gamma(\alpha) \geq n}$, and $\varphi_{\sigma(\alpha) \geq n}$. Then, the following result ensures an important property of distinguishability formulas.

\begin{restatable}{prop}{booleanCombPoly}\label{booleanCombPoly}
	Given a string $w$ and $\psi$ a boolean combination of distinguishability formulas, one can check if $w \models \psi$ in polynomial time.
\end{restatable}

\begin{proof}
	We need to show that given a string $w$ and a distinguishability formula $\varphi$, we can check if $w \models \varphi$ in polynomial time. From that, it follows directly that it also holds for any boolean combination of such formulas.
	
	First, let $\varphi = \varphi^{d(min, max)}_{\leq n}$. It is possible to check in linear time in the size of $w$ whether $|w| \leq n + 1$. 
	
	Next, let $\varphi = \varphi_{pref_k = a_1...a_k}$. Clearly, it is also possible to check in linear time in the size of $w$ whether $pref_k(w) = a_1...a_k$. The case in which $\varphi = \varphi_{suff_k = a_1...a_k}$ is analogous.
	
	Now, let $\varphi = \varphi_{\gamma(\alpha) \geq n}$. At each position of $w$, it is necessary to check if $\alpha$ occurs free and centered in that position. Then, it takes polynomial time in the size of $w$.
	
	Finally, let $\varphi = \varphi_{\sigma(\alpha) \geq n}$. First, it is necessary to compute a $2^{q_{\alpha}}$-segmentation of the set of free occurrences of $\alpha$. This takes polynomial time in the size of $w$. Then, it suffices to compare the number of partitions with $n$. Thus, it takes polynomial time to check whether $w \models \varphi_{\sigma(\alpha) \geq n}$.
	
	Therefore, if $\psi$ is a boolean combination of distinguishability formulas, then it takes polynomial to check whether $w \models \psi$.
\end{proof}

Distinguishability formulas are also representative for the set of first-order sentences over successor string structures. For example, let $\varphi = \exists x (P_a(x) \wedge \forall y (x \neq y \to \neg P_a(y)))$. Also let $u = bbabb$, $v_1 = bbbbb$, $v_2 = bbabbabb$, and $r = 2$. Therefore, $\varphi_{\gamma(a) \geq 1} \in \Phi^r_{u, v_1}$ and $\varphi_{\gamma(a) < 2} \in \Phi^r_{u, v_2}$. Thus, $\varphi$ is equivalent to $\varphi_{\gamma(a) \geq 1} \wedge \varphi_{\gamma(a) < 2}$. Now we will show that this holds for any first-order sentence over strings in our setting. First, we define formulas equivalent to Hintikka formulas.

$$\varphi^{r, length}_{w}:=\left\{\begin{array}{ll}
\varphi^{d(min,max)}_{= d(min^{w}, max^{w})},&\textrm{ if }|w| \leq 2^r + 1\\
\varphi^{d(min, max)}_{> 2^r},&\textrm{ otherwise.}
\end{array}\right.$$
$$\varphi^{r, pref}_{w}:= \varphi_{pref_{2^r} = pref_{2^r}(w)}$$
$$\varphi^{r, suff}_{w}:= \varphi_{suff_{2^r} = suff_{2^r}(w)}$$
$$\varphi^{r, \alpha}_{w}:=\left\{\begin{array}{ll}
\varphi_{\sigma(\alpha) = \sigma(\alpha, w)} \wedge \varphi_{\gamma(\alpha) = \gamma(\alpha, w)},&\textrm{ if }q_{\alpha} + \sigma(\alpha, w) \leq r\\
\varphi_{\sigma(\alpha) \geq r - q_{\alpha} + 1},&\textrm{ otherwise.}
\end{array}\right.$$
$$\varphi^{r, sub}_{w}:= \bigwedge \{\varphi^{r, \alpha}_{w} \mid |\alpha| = 2^q - 1, q > 0 \}.$$

\begin{restatable}{lem}{hintikkaStrings}\label{hintikkaStrings}
	$\models \varphi^{r}_{w} \leftrightarrow (\varphi^{r, length}_{w} \wedge \varphi^{r, pref}_{w} \wedge \varphi^{r, suff}_{w} \wedge \varphi^{r, sub}_{w})$.
\end{restatable}

\begin{proof}
	Let $u \models \varphi^{r}_{w}$. Then, the Duplicator has a winning strategy in $\mathcal{G}_r(w, u)$ and the conditions of Theorem~\ref{EFGamesOnStrings} hold. Then, $d(min^w, max^w) = d(min^u, max^u)$ or $d(min^w, max^w) > 2^r$ and $d(min^u, max^u) > 2^r$. We have two cases depending on the size of $w$:\\
	1. $|w| \leq 2^r + 1$. Then, $d(min^w, max^w) = d(min^u, max^u)$ and it follows that $u \models \varphi^{r, length}_{w}$.\\
	2. $|w| > 2^r + 1$. Then, $d(min^w, max^w) > 2^r$ and $d(min^u, max^u) > 2^r$. Therefore, $u \models \varphi^{r, length}_{w}$.
	
	From Theorem~\ref{EFGamesOnStrings}, it also holds that $pref_{2^r}(w) = pref_{2^r}(u)$ and $suff_{2^r}(w) = suff_{2^r}(u)$. Then, $u \models \varphi^{r, pref}_{w} \wedge \varphi^{r, suff}_{w}$.
	
	From condition~\ref{gameSub} of Theorem~\ref{EFGamesOnStrings}, it holds that $\sigma(\alpha, w) + q_{\alpha} > r$ and $\sigma(\alpha, u) + q_{\alpha} > r$ for all $\alpha$ such that $|\alpha| = 2^{q_{\alpha}} - 1$ and $\sigma(\alpha, w) \neq \sigma(\alpha, u)$ or $\gamma(\alpha, w) \neq \gamma(\alpha, u)$. Let $\alpha$ such that $|\alpha| = 2^q - 1$. We have two cases:\\
	1. $q + \sigma(\alpha, w) \leq r$. Then, $\sigma(\alpha, w) = \sigma(\alpha, u)$ or $\gamma(\alpha, w) = \gamma(\alpha, u)$. Therefore, $u \models \varphi^{r, \alpha}_w$.\\
	2. $q + \sigma(\alpha, w) > r$. If $\sigma(\alpha, u) + q \leq r$, then $\sigma(\alpha, w) \neq \sigma(\alpha, u)$. Therefore, $\sigma(\alpha, u) + q > r$. Hence, $u \models \varphi^{r, \alpha}_w$.
	Finally, $u \models \varphi^{r, sub}_{w}$.\\
	
	Conversely, let $u \models \varphi^{r, length}_{w} \wedge \varphi^{r, pref}_{w} \wedge \varphi^{r, suff}_{w} \wedge \varphi^{r, sub}_{w}$. We have that $pref_{2^r}(w) = pref_{2^r}(u)$ and $suff_{2^r}(w) = suff_{2^r}(u)$ because $u \models \varphi^{r, pref}_{w} \wedge \varphi^{r, suff}_{w}$. We also have that $u \models \varphi^{r, length}_{w}$. We have two cases:\\
	1. $|w| \leq 2^r + 1$. Then, $d(min^w, max^w) = d(min^u, max^u)$.\\
	2. $|w| > 2^r + 1$. Then, $d(min^w, max^w) > 2^r$ and $d(min^u, max^u) > 2^r$.
	
	\noindent It also holds that $u \models \varphi^{r, sub}_{w}$. Let $\alpha$ such that $|\alpha| = 2^q - 1$ and $\sigma(\alpha, w) \neq \sigma(\alpha, u)$ or $\gamma(\alpha, w) \neq \gamma(\alpha, u)$. Then, $u \models \varphi^{r, \alpha}_w$. We have that $\sigma(\alpha, w) + q > r$, otherwise, $\sigma(\alpha, w) = \sigma(\alpha, u)$ and $\gamma(\alpha, w) = \gamma(\alpha, u)$. Therefore, $\sigma(\alpha, u) + q > r$. Then, for all $\alpha$ such that $|\alpha| = 2^q - 1$ and $\sigma(\alpha, w) \neq \sigma(\alpha, u)$ or $\gamma(\alpha, w) \neq \gamma(\alpha, u)$, we have that $\sigma(\alpha, w) + q > r$ and $\sigma(\alpha, u) + q > r$. All the conditions of Theorem~\ref{EFGamesOnStrings} hold. Then, the Duplicator has a winning strategy in $\mathcal{G}_r(w, u)$. It follows that $u \models \varphi^r_w$.
\end{proof}

Now, we need the following lemmas.

\begin{lem}\label{lemaLength}
	Let $r$ be a natural number and $w$ a string. There is a set of strings $V_{length}$ such that $\varphi^{r, length}_w$ is equivalent to a boolean combination of sentences in $\bigcup_{v \in V_{length}} \Phi^r_{w, v}$.
\end{lem}

\begin{proof}
	If $\varphi^{r, length}_w = \varphi^{d(min, max)}_{> 2^r}$, then let $V_{length} = \{v\}$ such that $|v| = 2^r + 1$. Then, $d(min^{v}, max^{v}) = 2^r$. Observe that $\models \varphi^{r, length}_w \leftrightarrow \varphi^{d(min, max)}_{\geq 2^r + 1}$ and $\varphi^{d(min, max)}_{\geq 2^r + 1} \in \Phi^r_{w, v}$ because $|w| > |v|$ and $|v| \leq 2^r + 1 \leq min(2^r + 1, |u| - 1)$. If $\varphi^{r, length}_w = \varphi^{d(min, max)}_{= d(min^w, max^w)}$, then let $V_{length} = \{ v_1, v_2 \}$ such that $|v_1| = |w| - 1$ and $|v_2| = |w| + 1$. Thus, $|v_1| < |w|$ and $|v_1| \leq d(min^w, max^w) \leq |w| - 1$. It follows that $\varphi^{d(min, max)}_{\geq d(min^w, max^w)} \in \Phi^r_{w, v_1}$. We also have that $|v_2| > |w|$ and $|w| - 1 \leq d(min^w, max^w) \leq |v_2| - 2$. Therefore, $\varphi^{d(min, max)}_{\leq d(min^w, max^w)} \in \Phi^r_{w, v_2}$. Obviously, $\varphi^{r, length}_w$ is equivalent to $\varphi^{d(min, max)}_{\geq d(min^w, max^w)} \wedge \varphi^{d(min, max)}_{\leq d(min^w, max^w)}$.
\end{proof}

\begin{lem}\label{lemaprefsuff}
	Let $r$ be a natural number and $w$ a string. There is a set of strings $V_{prefsuff}$ such that $\varphi^{r, pref}_w \wedge \varphi^{r, suff}_w$ is equivalent to a boolean combination of sentences in $\bigcup_{v \in V_{prefsuff}} \Phi^r_{w, v}$.
\end{lem}

\begin{proof}
	Let $v_1$ such that $|v_1| = |w|$ and $v_1 \neq w$. If $|w| \leq 2^r$, then it follows that $\models \varphi_{pref_{2^r} = pref_{2^r}(w)} \leftrightarrow \varphi_{pref_{|w|} = w}$. Hence, $\varphi_{pref_{|w|} = w} \in \Phi^r_{w, v_1}$. If $|w| > 2^r$, then $\varphi_{pref_{2^r} = pref_{2^r}(w)} \in \Phi^r_{w, v_1}$. Obviously, the case for $\varphi_{suff_{2^r} = suff_{2^r}(w)}$ is analogous. Let $v_2$ such that $\varphi_{suff_{2^r} = suff_{2^r}(w)} \in \Phi^r_{w, v_2}$. Therefore, $V_{prefsuff} = \{ v_1, v_2\}$.
\end{proof}

\begin{lem}\label{lemaalpha}
	Let $r$ be a natural number and $w, \alpha$ strings. There is a set of strings $V_{\alpha}$ such that $\varphi^{r, \alpha}_w$ is equivalent to a boolean combination of sentences in $\bigcup_{v \in V_{\alpha}} \Phi^r_{w, v}$.
\end{lem}

\begin{proof}
	If $\varphi^{r, \alpha}_w = \varphi_{\sigma(\alpha) \geq r - q_{\alpha} + 1}$, then $\sigma(\alpha, w) \geq r - q_{\alpha} + 1$. Let $V_{\alpha} = \{ v \}$ such that $\sigma(\alpha, w) > \sigma(\alpha, v)$ and $\sigma(\alpha, v) < r - q_{\alpha} + 1$. Thus, $\varphi_{\sigma(\alpha) \geq r - q_{\alpha} + 1} \in \Phi^r_{w, v}$. If $\varphi^{r, \alpha}_w = \varphi_{\sigma(\alpha) = \sigma(\alpha, w)} \wedge \varphi_{\gamma(\alpha) = \gamma(\alpha, w)}$, then $\sigma(\alpha, w) \leq r - q_{\alpha}$. For $\varphi_{\sigma(\alpha) = \sigma(\alpha, w)}$, let $v_1$ such that $\sigma(\alpha, w) > \sigma(\alpha, v_1)$. Then, $\varphi_{\sigma(\alpha) \geq \sigma(\alpha, w)} \in \Phi^r_{w, v_1}$. Let $v_2$ such that $\sigma(\alpha, v_2) > \sigma(\alpha, w)$. Note that $\sigma(\alpha, w) < \sigma(\alpha, w) + 1 \leq min(r - q_{\alpha} + 1, \sigma(\alpha, v_2))$. Then, $\varphi_{\sigma(\alpha) < \sigma(\alpha, w) + 1} \in \Phi^r_{w, v_1}$. Clearly, the case for $\varphi_{\gamma(\alpha) = \gamma(\alpha, w)}$ is analogous. Let $v_3$ and $v_4$ such that $\varphi_{\gamma(\alpha) \geq \gamma(\alpha, w)} \in \Phi^r_{u, v_3}$ and $\varphi_{\gamma(\alpha) < \gamma(\alpha, w) + 1} \in \Phi^r_{u, v_4}$. Therefore $V_{\alpha} = \{v_1, v_2, v_3, v_4 \}$.
\end{proof}

\begin{restatable}{lem}{hintikkaSample}\label{hintikkaSample}
	Let $r$ be a natural number and $w$ a string. There is a set of strings $V$ such that $\varphi^{r}_{w}$ is a boolean combination of sentences in $\bigcup_{v \in V} \Phi^r_{w, v}$.
\end{restatable}

\begin{proof}
	By Lemma~\ref{hintikkaStrings}, $\models \varphi^r_w \leftrightarrow (\varphi^{r, length}_{w} \wedge \varphi^{r, pref}_{w} \wedge \varphi^{r, suff}_{w} \wedge \varphi^{r, sub}_{w})$. Let $V = \bigcup_{ \{\alpha \mid |\alpha| = 2^q - 1, q > 0\} }$ $V_{\alpha}$ $\cup$ $V_{length}$ $\cup$ $V_{prefsuff}$ $\cup$  as in Lemma~\ref{lemaLength}, Lemma~\ref{lemaprefsuff}, and Lema~\ref{lemaalpha}. From these lemmas, it follows that, $\varphi^r_w$ is equivalent to a boolean combination of sentences in $\bigcup_{v \in V} \Phi^r_{w, v}$.
\end{proof}

The next result is related to Theorem~\ref{hintikkaRepresentative}. It suggests that our approach is likely to find any first-order sentence given a suitable sample of strings. Recall that, by Theorem~\ref{hintikkaRepresentative}, any first-order sentence is equivalent to a disjunction of Hintikka formulas. Thus, we have the following result.

\begin{restatable}{teo}{algoRobust}
	Let $\varphi$ be a first-order sentence over successor string structures. Then, $\varphi$ is equivalent to a boolean combination of distinguishability formulas.
\end{restatable}

\begin{proof}
	Let $r$ such that $qr(\varphi) = r$. From Theorem~\ref{hintikkaRepresentative} it follows that $\models \varphi \leftrightarrow \varphi^r_{u_1} \vee ... \varphi^r_{u_s}$. Let $U = \{u_1, ..., u_s\}$. In accord to Lemma~\ref{hintikkaSample}, let $V_i$ be such that $\varphi^r_{u_i}$ is equivalent to a boolean combination of sentences in $\bigcup_{v \in V_i} \Phi^r_{u_i, v}$. Therefore, the sentence $\varphi$ is equivalent to a boolean combination of sentences in $\bigcup^s_{i = 1} (\bigcup_{v \in V_i} \Phi^r_{u_i, v})$.
\end{proof}

\section{The Algorithm and Its Analysis}\label{sec4}

In this section, we define an algorithm for finding a first-order sentence $\varphi_S$ from a sample of strings $S$. Subformulas of $\varphi_S$ are distinguishability formulas from sets of the form $\Phi^r_{u, v}$ such that $u \in P$ and $v \in N$. We also give an example of how the algorithm works. We guarantee that our algorithm runs in polynomial time in the size of the input sample $S$. The size of the sample $S$ is the sum of the lengths of all strings it includes. We use $|S|$ to denote the size of the sample $S$. We also show that $\varphi_S$ returned by our algorithm is consistent with $S$. Furthermore, we also prove that $\varphi_S$ is a sentence of minimal quantifier rank consistent with $S$. The pseudocode of our algorithm is in Algorithm~\ref{alg:distinguir}.

\begin{algorithm}%[H]
	\caption{}\label{alg:distinguir}
\begin{algorithmic}%[1]
\State \textbf{Input:} Sample of strings $S = (P, N)$
\State $r \gets max\{EFsim(u, v) \mid u \in P, v \in N\}$
\State $\varphi_S \gets \bigvee_{u \in P} \bigwedge_{v \in N}$ \textbf{choose} $\varphi \in \Phi^r_{u, v}$
\State \textbf{return} $\varphi_S$

\end{algorithmic}
\end{algorithm}

First, the algorithm finds the minimum value $r$ such that there exists a sentence of quantifier rank $r$ that is consistent with the input sample $S$. After that, the algorithm constructs $\varphi_S$. It goes through all strings in $P \cup N$, and, for $u \in P, v \in N$, it chooses a formula $\varphi \in \Phi^r_{u, v}$. For each $u \in P$, Algorithm~\ref{alg:distinguir} builds a conjunction of sentences in $\bigcup_{v \in N} \Phi^r_{u, v}$. Finally, it returns a disjunction of such conjunctions. In the following, we show an example of how this algorithm works on a simple instance.

\begin{exam}\label{exampleAlgo}
	Let $S$ be the sample in Table~\ref{tableExample1}. Note that $max \{EFsim(u,v) \mid u \in P, v \in N\} = 1$ as witnessed by $\sigma(p,stviie) + 1 \leq 1$ and $\sigma(p,stpiie) + 1 > 1$. Clearly, $\varphi_{pref_1 = s} \in \Phi^{1}_{stviil, ktvive}$, $\varphi_{suff_1 \neq e} \in \Phi^{1}_{stviil, stpiie}$, $\varphi_{pref_1 = s} \in \Phi^{1}_{stviie, ktvive}$, and $\varphi_{\sigma(p) < 1} \in \Phi^{1}_{stviie, stpiie}$. Therefore, Algorithm~\ref{alg:distinguir} returns $\varphi_S$ below. Observe that $\varphi_S$ is consistent with $S$ and $qr(\varphi_S) = 1$.
	$$\varphi_S = (\varphi_{pref_1 = s} \wedge \varphi_{suff_1 \neq e}) \vee (\varphi_{pref_1 = s} \wedge \varphi_{\sigma(p) < 1}).$$
\end{exam}

\noindent In the following, we prove the correctness and the time complexity of our algorithm. First, we show that it returns a sentence that is consistent with the sample. After that, we show that it returns a sentence of minimal quantifier rank. Then, we prove that the running time of our learning algorithm is polynomial in the size of the given sample.

\begin{restatable}{teo}{algoCons}\label{consisData}
	Let $S$ be a sample and $\varphi_S$ returned by Algorithm~\ref{alg:distinguir}. Then, $\varphi_S$ is consistent with $S$.
\end{restatable}

\begin{proof}
	Let $u \in P$. Then, $\varphi_{u} = \varphi_1 \wedge ... \wedge \varphi_k$ such that, for all $i$, $\varphi_i \in \Phi^r_{u, v}$, for some $v \in N$. In this way, $u \models \varphi_i$, for all $i$ and then $u \models \varphi_S$. Now, let $v \in N$ and assume that $v \models \varphi_S$, i.e, $v \models \varphi_{u}$, for some $u \in P$. Therefore, $\varphi_{u}$ has a conjunct $\varphi \in \Phi^r_{u, v}$. This is an absurd because, from Lemma~\ref{setConsis}, it follows that $v \not\models \varphi$.
\end{proof}

\begin{restatable}{teo}{algoQR} 
	The sentence $\varphi_S$ returned by Algorithm~\ref{alg:distinguir} is a first-order sentence of minimal quantifier rank that is consistent with $S$.
\end{restatable}

\begin{proof}
	Suppose a first-order sentence $\psi$ consistent with $S$ such that $qr(\psi) < qr(\varphi_S) = max\{EFsim(u, v) \mid u \in P, v \in N\}$. Let $u' \in P$ and $v' \in N$ such that $EFsim(u', v') = max\{EFsim(u, v) \mid u \in P, v \in N\}$. Then, $u'$ and $v'$ are satisfied by the same first-order sentences of quantifier rank $q$ such that $q < EFsim(u', v')$. Then, $u' \models \psi$ iff $v' \models \psi$. Therefore, $\psi$ is not consistent with $S$. This is a contradiction.
\end{proof}

\begin{restatable}{teo}{hypPoly}\label{algoPoly}
	Given a sample $S$, Algorithm~\ref{alg:distinguir} returns $\varphi_S$ in time $O(|S|^7)$.
\end{restatable}

\begin{proof}
	First, the algorithm computes $max \{EFsim(u,v) \mid u \in P, v \in N \}$ in order to use a suitable quantifier rank. This takes time $O(|S|^4 \times log(|S|))$ because, for a given $u \in P, v \in N$, $|u| + |v| < |S|$, to compute $EFsim(u, v)$ takes time $O(|S|^2 \times log(|S|))$, and this procedure is executed $|S|^2$ times. Then, our algorithm loops over strings in the sample and, in each loop, it chooses a formula $\varphi \in \Phi^r_{u, v}$. As the size of each $\varphi \in \Phi^r_{u,v}$ is $O(|S|^2)$ and $|\Phi^r_{u, v}|$ is $O(|S|^3)$, one iteration of the loop runs in time $O(|S|^5)$. This loop is executed $O(|S|^2)$ times, then, this loop takes time $O(|S|^7)$. The first step runs in time $O(|S|^4log(|S|))$ and the rest takes time $O(|S|^7)$. Therefore the overall complexity of Algorithm~\ref{alg:distinguir} is $O(|S|^7)$.
\end{proof}

Therefore, our algorithm is an improvement over the work in \cite{kaiser12lbg}, for this particular problem on successor string structures. However, observe that the sentence returned by Algorithm~\ref{alg:distinguir} is a disjunction of $|P|$ conjunctions of distinguishability formulas. The algorithm in \cite{kaiser12lbg} also returns formulas which are long and hard to read. Then, they greedily remove subformulas that are not necessary. Now, we also consider the number of conjunctions in our approach. Let $\Phi_S := \bigcup_{u \in P, v \in N} \Phi^{max\{EFsim(u, v) \mid u \in P, v \in N\}}_{u, v}$. We say that a first-order formula is in $m$-DDF (Disjunctive Distinguishability Form) over $\Phi_S$ if it is a disjunction of $m$ conjunctions of distinguishability formulas in $\Phi_S$. Therefore, we can also define a problem where the goal is to find a formula $\varphi_S$ in $m$-DDF such that $m$ is minimum. This formula $\varphi_S$ improves interpretability. Given a sample $S$, the problem consists of finding a first-order sentence $\varphi_S$ in $m$-DDF over $\Phi_S$ such that $\varphi_S$ is consistent with $S$, and $m$ is minimum. 

An algorithm to return a first-order sentence given a sample of strings and a set of first-order sentences is presented in \cite{dnfstrings2018}. Formally, given a sample of strings $S$ and a set of first-order sentences $\Phi$, the goal is to find a first-order sentence $\varphi$ such that $\varphi$ is consistent with $S$, $\varphi$ is a disjunction of conjunctions of sentences in $\Phi$, and the number of conjunctions is minimum. This problem is NP-complete. It is easy to polynomially reduce our problem of distinguishability formulas to this problem in \cite{dnfstrings2018}. Then, our problem of distinguishability formulas is in NP, and it is still a better approach than the one in \cite{kaiser12lbg}, for the particular case of successor string structures. In this new setting, the formula for the sample of Table~\ref{tableExample1} is below. This formula is smaller than the formula of Example~\ref{exampleAlgo}.
$$\varphi_S = \varphi_{pref_1 = s} \wedge \varphi_{\gamma(v) \geq 1}.$$

\section{Conclusions and Future Work}\label{sec5}

Motivated by the framework defined in \cite{kaiser12lbg} and results of the Ehrenfeucht–Fraïssé Game over successor string structures in \cite{montanari05}, we introduced an algorithm that returns a first-order sentence of minimal quantifier rank that is consistent with a given sample of strings. Our algorithm runs in time $O(|S|^7)$ where $|S|$ is the size of the input sample. Then, our algorithm runs in polynomial time in the size of $S$. The algorithm in \cite{kaiser12lbg} runs in exponential time as it works for arbitrary structures.

We designed our algorithm using distinguishability formulas which are defined based on the conditions characterizing a winning strategy for the Spoiler on successor string structures. Our algorithm returns a disjunction of conjunctions of distinguishability formulas. The size of a distinguishability formula is polynomial in the length of two given strings. The algorithm in \cite{kaiser12lbg} uses Hintikka formulas which have size exponential in the size of a given string. Therefore, our proposed algorithm is an improvement over the one in \cite{kaiser12lbg}, for successor string structures. We also show that any first-order sentence is equivalent to a boolean combination of distinguishability formulas. Then, our approach has the potential to find any first-order sentence given the right sample. We also showed how to find a formula with the minimum number of conjunctions. A small formula is preferable for explaining a sample of strings.

Our results are also relevant to grammatical inference, where the goal is to find a language model of minimal size that describes a given sample of strings. In our framework, the language model is the first-order logic over successor string structures, and we use the quantifier rank as a natural measure of first-order sentences. As strings may be used to model sequences of symbolic data, our results can be applied in the analysis of biological sequences \cite{wieczorek2014induction}. A recent model-theoretic approach to grammatical inference \cite{strother2017using} uses a fragment of first-order logic which is less expressive than full first-order logic in our approach.

As future work, we intend to explore the problem of finding a first-order sentence over strings with the linear order relation. First-order logic over strings with the linear order defines the class of star-free languages \cite{thomas1982classifying}, and it is more expressive than first-order logic over successor string structures. Finally, we plan to extend our approach to monadic second-order logic. Regular languages are exactly the languages definable in monadic second-order logic \cite{Buchi60}. An algorithm which returns monadic second-order sentences can be used in the problem of finding a finite automaton consistent with a given sample of strings.

%\nocite{*}
\bibliographystyle{eptcs}
\bibliography{references}
\end{document}